\documentclass[runningheads]{llncs}
\usepackage[T1]{fontenc}
\usepackage{graphicx}
\def \VersionWithComments {}

\ifdefined\VersionWithComments
\usepackage[colorinlistoftodos,textsize=footnotesize]{todonotes}
\setuptodonotes{inline} 
\else
\usepackage[disable]{todonotes}
\fi

\usepackage{amsmath}
\usepackage{amssymb}
\usepackage{stmaryrd}
\usepackage{algorithm}
\usepackage{algpseudocode}
\usepackage{xspace}
\usetikzlibrary{calc,fit}
\definecolor {processblue}{cmyk}{0.96,0,0,0}

\DeclareMathOperator{\lfp}{lfp}
\newcommand{\undet}{\ensuremath{\rlap{$\top$}\bot}}
\newcommand{\threeval}{\textbf{3}}

\spnewtheorem{notation}{Main Theorem}{\itshape}{\rmfamily}

\newcommand{\diam}[1]{\Diamond_{#1}\,}
\newcommand{\F}[1]{F_{#1}\,}

\newcommand{\chain}[1]{\textmd{chain}_{#1}}

\newcommand{\sem}[3]{\ensuremath{{\left\llbracket{#1}\right\rrbracket_{#2}({#3})}}}
\newcommand{\semExt}[3]{\ensuremath{{\left\llbracket{#1}\right\rrbracket_{#2}^{e}({#3})}}}
\newcommand{\semfun}[2]{\ensuremath{{\left\llbracket{#1}\right\rrbracket_{#2}}}}
\newcommand{\semextfun}[2]{\ensuremath{{\left\llbracket{#1}\right\rrbracket_{#2}^{e}}}}

\newcommand{\rsemextfunbis}[2]{\ensuremath{{\left\llbracket{#1}\right\rrbracket_{#2}^{c}}}}
\newcommand{\semwofp}[5]{\ensuremath{{\left\llbracket{#1}\right\rrbracket_{#2}({#3,#4,#5})}}}

\newcommand{\bigland}{\bigwedge}
\newcommand{\biglor}{\bigvee}

\newcommand{\fvar}{\ensuremath{\varphi}}

\newcommand{\spaceHorizon}[1]{\ensuremath{\bullet_{#1}}}
\newcommand{\spb}{\: | \:}
\newcommand{\mutgl}{$\mu$-TGL\xspace}

\newcommand{\Reach}[4]{\ensuremath{#3 \,\mathcal{R}^{#1}_{#2}\, #4}}
\newcommand{\Escape}[4]{\ensuremath{#3 \,\mathcal{E}^{#1}_{#2}\, #4}}

\newcommand{\map}[3]{#1\:\colon\:#2\to#3}
\newcommand{\agents}{\mathcal{A}}
\newcommand{\reals}{\mathbb{R}}
\newcommand{\ereals}{\overline{\mathbb{R}}}
\newcommand{\preals}{\mathbb{R}_{\geq 0}}
\newcommand{\tgraph}{\sigma}

\algrenewcommand\algorithmicrequire{\textbf{Input:}}
\algrenewcommand\algorithmicensure{\textbf{Output:}}

\begin{document}

	\title{Monitoring Diameters of Causal Communication Graph with Spatio-Temporal Logic}
	\titlerunning{Monitoring Diameters of Causal Communication Graph}
	%
	\author{Lydia Bakiri\orcidID{0009-0003-8655-6455} \and
	Jérémy Dubut\orcidID{0000-0002-2640-3065} \and
	Sergio Mover\orcidID{0000-0003-1029-9547} 
	} 
	\authorrunning{L. Bakiri et al.}
	%
	\institute{LIX, CNRS, École polytechnique, Institut Polytechnique de Paris, Palaiseau, France\\ 
	\email{firstname.lastname@polytechnique.edu}
	}
	\maketitle          
	
	\begin{abstract}
		Verification of multi-agent systems requires the ability to check meticulous 
		topological properties when it comes to agents that can move through 
		space in continuous time. This demands a logic with sufficient 
		expressiveness to capture these dynamics. 
		\mutgl logic has interesting properties for expressing entangled 
		space-time properties. However, this logic lacks the expressivity needed 
		to analyse reachability within specific distance bounds, 
		or to track the length or the cost of communication chains: 
		these are fundamental for decentralized monitoring, or graph-theoretic 
		analysis of distributed protocols, where algorithmic complexities 
		often relates with the 
		system's communication graph diameter. We then introduce an extension of 
		\mutgl, including a new operator called the space horizon. This 
		addition allows us to bound the distance of communication chains, 
		hence enhancing the logic's expressiveness. We show that this operator 
		allows to encode modalities from other logics, such as reachability or 
		escaping which were not available in vanilla \mutgl, while allowing a deeper 
		entanglement of spatial and temporal properties. We provide a centralized 
		offline monitoring algorithm for this logic and illustrate it on several examples 
		on simulations of Consensus-Based Bundle Algorithms,
		distributed protocols for task allocation.
	\end{abstract}
	
	\section{Introduction}

	Increasingly, critical infrastructures are being deployed as systems of systems 
	that are highly distributed, autonomous, and interconnected. 
	Such systems are becoming progressively more complex due to the growing number 
	of interacting components and their dynamic environments. 
	As a consequence, providing strong assurances regarding their correctness, 
	safety, and overall behavior is becoming increasingly challenging. 
	A central concern is the robust and efficient realization of distributed 
	coordination mechanisms, such as consensus, task assignment, rendez-vous 
	protocols, and related collective behaviors. 
	These mechanisms must remain operational in the presence of failures or 
	unexpected disruptions while simultaneously minimizing communication overhead, 
	execution time, energy consumption, and other resource costs. 
	Achieving these objectives strongly depends on the structure of the underlying 
	communication graph. 
	Robustness requirements often demand sufficient connectivity to tolerate 
	failures and preserve information flow, whereas efficiency considerations 
	require properties such as a limited graph diameter to reduce protocol rounds 
	and resource consumption. 
	The problem becomes significantly more challenging when communication 
	topologies evolve over time, as occurs naturally in systems whose components 
	move through physical space, such as fleets of autonomous vehicles or swarms 
	of drones. 
	In such settings, desirable graph properties become intricate entanglement 
	of spatial and temporal constraints. 
	Their analysis requires particular care because causality effects emerge 
	naturally: communication relations are inherently directed in time, and 
	chains of interactions cannot be treated as static connections. 
	For example, if subsystem A communicates with B at time $t_0$, and B 
	subsequently communicates with C at time $t_1 > t_0$, information may flow 
	from A to C, but not necessarily in the reverse direction. 
	Temporal ordering therefore induces asymmetric information dependencies 
	that must be explicitly accounted for when reasoning about correctness and 
	performance.

	Given the complexity and dynamic nature of these systems, applying heavyweight 
	formal verification techniques at the scale of realistic systems of systems 
	quickly becomes impractical. 
	Instead, we focus on lightweight assurance methods, and in particular on 
	monitoring techniques. 
	Our setting is especially well-suited to centralized monitoring architectures, 
	as many real-world systems of systems involve a human operator responsible for 
	high-level strategic decisions while the individual subsystems operate 
	autonomously at a lower level, 
	as it is the case for \emph{directed} systems of systems~\cite{maier98}. 
	In this context, online monitoring can serve several complementary purposes. 
	It may provide feedback to assist the human operator in adapting or correcting 
	global strategies whenever potentially unsafe communication patterns or 
	undesirable emergent behaviors are detected. 
	It can also be integrated directly into a control loop, where it 
	acts as a runtime filtering mechanism capable of rejecting 
	decisions or execution paths predicted to lead to communication 
	failures or degraded performance, as a safety shield 
	during the learning of optimal policies~\cite{MA18}, or as a 
	component of a decision module responsible for selecting the
	appropriate controller within a Simplex architecture~\cite{DS98}.
	However, enabling such capabilities requires addressing two fundamental 
	challenges. 
	First, one needs a specification language expressive enough to capture the 
	intricate interplay between spatial and temporal communication properties, 
	including the causality constraints induced by evolving interaction patterns. 
	Second, one needs effective monitoring algorithms capable of evaluating these 
	specifications.

	In this work, we consider \mutgl (Timed Graph Logic with fixpoint), a 
	spatio-temporal logic previously introduced in \cite{ER22}. 
	This logic combines ideas from Metric Interval Temporal Logic (MITL~\cite{AH94}) and the 
	modal $\mu$-calculus~\cite{DK83}. It includes temporal modalities similar to the MITL F 
	operator, as well as spatial modalities expressing the existence of agents 
	satisfying a given property in the vicinity of the current agent. 
	Crucially, \mutgl also incorporates fixpoint operators. 
	While the temporal and spatial modalities themselves are relatively local in 
	nature—contrary to more global operators such as temporal until or graph 
	reachability—the addition of fixpoints makes it possible to express 
	complex recursive patterns that tightly interleave spatial and temporal 
	reasoning. 
	Such expressive power is essential for capturing the kinds of properties 
	discussed above, where information propagation depends simultaneously on 
	connectivity structures, movement patterns, and causality constraints. 
	In particular, this framework enabled the specification of notions such as 
	causal robust connectivity, a property extending robust connectivity,
	required to guarantee consensus in 
	the presence of faulty agents~\cite{LB13}, to the context of dynamically evolving communication graphs. 
	A major challenge in this setting lies in designing semantics suitable for 
	online monitoring. Following standard approaches, we rely on a three-valued 
	semantics~\cite{MB12} in which an indeterminate value represents situations where the 
	available observations are still insufficient to decide whether a formula is 
	satisfied or violated. 
	However, the presence of fixpoint operators introduces significant 
	additional difficulties. 
	In particular, care is required to manage computationally the 
	notion of a time horizon, representing the amount of future information that 
	may be needed before safely concluding that a property cannot be satisfied.

	In this paper, we address a limitation of the current logic. 
	While \mutgl allows one to constrain temporal propagation through its notion 
	of time horizon, it currently lacks the ability to express that information 
	propagates through communication chains that are spatially bounded, namely 
	chains whose length remains limited in terms of communication hops or spatial 
	cost. 
	Such a capability is important in several related frameworks. 
	For instance, the STREL logic~\cite{LN20,LN22} introduces more global spatio-temporal operators 
	such as reachability and escape, which have been successfully used to model 
	phenomena including network routing protocols and virus propagation. 
	In our context, this type of property is equally important for the 
	monitoring of distributed protocols. 
	The ability to constrain spatial propagation would make it possible to 
	specify and monitor properties such as bounded communication diameter, 
	which directly impacts the number of rounds required by many distributed 
	algorithms, or bounded communication cost, which is closely related to energy 
	consumption and resource usage. 
	The syntactic extension proposed in this paper is remarkably simple: we 
	introduce a space horizon operator that mirrors the role played by the 
	time horizon mechanism in standard \mutgl. However, extending the semantics 
	and monitoring procedures to support this additional operator is substantially 
	more challenging. 
	The space budget must be carefully propagated and managed throughout the 
	evaluation process, while the presence of fixpoint operators again 
	introduces significant technical difficulties. 
	In particular, ensuring that fixpoints remain computable 
	through a finite number of iterations requires substantial modifications to 
	both the semantic framework and the monitoring algorithm.

	\subsection*{Related Work}
	\textbf{(Spatio-)Temporal Logics.}
	Our work is based on \mutgl~\cite{ER22}. The closest logic 
	is STREL~\cite{LN20,LN22}
	-- and by extension, its ancestors 
	Spatial-Temporal Logic~\cite{IH15} (SpaTeL) and
	Signal Spatio-Temporal 
	Logic~\cite{LN18} (SSTL)
	-- however they cannot express suitable 
	causal properties because their modalities are not sufficiently 
	fine-grained in their interactions between time and space: their space modalities 
	only look at the communication graph at a fixed time. 
	Similarly, Counting Linear Temporal Logic~\cite{YES20} (cLTL) enables the 
	specification of desired behaviors in multi-agent systems by 
	allowing constraints on the number of agents satisfying particular 
	LTL properties. 
	The recent Hybrid Spatio-temporal Logic~\cite{RFT26} (HSTL) 
	combines ideas from both spatio-temporal and 
	hybrid logics and is targeting the safety assurance of 
	fleets of autonomous vehicle in highway scenarios.
	Although Signal Temporal Logic (STL) is not, 
	strictly speaking, a spatio-temporal logic, spatial aspects can 
	nevertheless be encoded within its atomic propositions, allowing 
	certain spatial properties to be represented indirectly.
	More generally, logics for the verification of multi-agent 
	systems is a hot topic, see for example Alternating-Time 
	Temporal Logic~\cite{RA02} and its various extensions and fragments.
	
	\textbf{Distributed Protocols.}
	The motivation for this work is to be able to monitor 
	spatio-temporal properties to ensure that some distributed protocols 
	can be executed, even when agents can fail in a malicious 
	way (Byzantine). The distributed protocols we are interested in 
	are those that allow the agents to agree: 
	consensus~\cite{MP80}, 
	task allocation typically through auctions~\cite{CH09}, 
	gathering~\cite{CC08}, etc.

	\subsection*{Contributions} We extend both 
	the syntax and the semantics of \mutgl, to be able 
	to specify properties on the diameter of 
	causal communication graphs or the length of communication chains.
	We implement an offline monitoring algorithm, 
	implementing this semantics, and carefully prove its 
	termination. We illustrate its feasibility on 
	simulations of agents following Consensus-Based 
	Bundle Algorithms, distributed protocols for task 
	allocation.

	\section{Recap on \mutgl}
	
\mutgl and the extension introduced in this paper are 
logics over \emph{timed graphs}, that is, signals whose 
values are weighted graphs. 
Such a graph represents a network of agents: 
its vertices correspond to the agents themselves, 
while an edge between two agents indicates that they 
can communicate. 
The weight associated with an edge captures 
quantitative information about the communication, 
such as the distance between the agents, 
the energy required to communicate, or a similar metric.
Fix a set $\agents$ of \emph{agents}. 
A timed graph is simply a function
$\map{\tgraph}{\preals \times \agents^2}{\ereals_{\geq 0}}$,
which maps a time and a pair of agents to the cost of
communication between them, 
with the value $+\infty$ indicating that communication 
is impossible.
Typically, agents move within a Euclidean space 
$\reals^n$, and the weights of a timed graph are 
obtained by applying a \emph{weighting function} to the 
distances between agents. Examples include:
1) the identity function, in which case the weights are the distances themselves;
2) the constant function equal to $1$, in which case the weights only encode the existence of a communication link (referred to as \emph{hops} in~\cite{LN22});
3) an energy function of the form $E(d)=d^\gamma$, where $\gamma$ is the path-loss exponent, whose value depends on the environment~\cite{TR02}.
The primary purpose of \mutgl is to express 
specifications such as: 
``it is possible to send a message to agent 
$a$ through a causal chain of communications 
within a given time bound.'' 
Let $\chain{a}$ denote such a formula. 
The semantic value of $\chain{a}$ depends on: 
1) the timed graph, 
2) the agent from which the message originates, 
3) the time at which the message is sent, and 
4) the available time budget.

\begin{example}[Horizons]
Let us illustrate this on this timed graph

\begin{center}
	\begin{tikzpicture}[
		every node/.style={font=\large},
		vertex/.style={circle,fill=black,minimum size=8pt,inner sep=0pt},
		scale=0.8
	]

	\node[vertex] (a) at (0,6) {};
	\node[left=.5cm] (la) at (a) {$a$};

	\node[vertex] (b) at (0,3.5) {};
	\node[vertex] (c) at (1,3.5) {};

	\draw[thick] (b)--(c);
	\node[above=0.15cm] at ($(b)!0.5!(c)$) {$1$};

	\node[left=.5cm] (lb) at (b) {$b$};
	\node[right=.5cm] (lc) at (c) {$c$};

	\node[vertex] (d) at (1,1) {};
	\node[vertex] (e) at (0,1) {};

	\draw[thick] (d)--(e);
	\node[above=0.15cm] at ($(d)!0.5!(e)$) {$1$};

	\node[right=.5cm] (ld) at (d) {$d$};
	\node[left=.5cm] (le) at (e) {$e$};

	\node[
		draw,
		dashed,
		rounded corners,
		fit=(la)(le)(ld)(lc),
		inner sep=0.1cm
	] {};
	\end{tikzpicture}
	\qquad\qquad
	\begin{tikzpicture}[
		every node/.style={font=\large},
		vertex/.style={circle,fill=black,minimum size=8pt,inner sep=0pt},
		scale=0.8
	]

	\node[vertex] (a) at (0,6) {};
	\node[left=.5cm] (la) at (a) {$a$};

	\node[vertex] (b) at (0,3.5) {};
	\node[vertex] (c) at (1,3.5) {};

	\draw[thick] (b)--(c);
	\node[above=0.15cm] at ($(b)!0.5!(c)$) {$1$};

	\node[left=.5cm] (lb) at (b) {$b$};
	\node[right=.5cm] (lc) at (c) {$c$};

	\node[vertex] (d) at (1,2.5) {};
	\draw[thick] (c)--(d);
	\node[right=0.15cm] at ($(c)!0.5!(d)$) {$1$};
	\node[right=.5cm] (ld) at (d) {$d$};

	\node[vertex] (e) at (0,1) {};
	\node[left=.5cm] (le) at (e) {$e$};

	\node[
		draw,
		dashed,
		rounded corners,
		fit=(la)(le)(ld)(lc),
		inner sep=0.1cm
	] {};
	\end{tikzpicture}
	\qquad\qquad
	\begin{tikzpicture}[
		every node/.style={font=\large},
		vertex/.style={circle,fill=black,minimum size=8pt,inner sep=0pt},
		scale=0.8
	]

	\node[vertex] (a) at (0,6) {};
	\node[left=.5cm] (la) at (a) {$a$};
	\node[vertex] (b) at (0,5) {};
	\node[left=.5cm] (lb) at (b) {$b$};
	\draw[thick] (a)--(b);
	\node[left=0.15cm] at ($(a)!0.5!(b)$) {$1$};

	\node[vertex] (c) at (1,3.5) {};
	\node[right=.5cm] (lc) at (c) {$c$};
	\node[vertex] (d) at (1,2.5) {};
	\draw[thick] (c)--(d);
	\node[right=0.15cm] at ($(c)!0.5!(d)$) {$1$};
	\node[right=.5cm] (ld) at (d) {$d$};

	\node[vertex] (e) at (0,1) {};
	\node[left=.5cm] (le) at (e) {$e$};

	\node[
		draw,
		dashed,
		rounded corners,
		fit=(la)(le)(ld)(lc),
		inner sep=0.1cm
	] {};

	\end{tikzpicture}
\end{center}
where the left graph is at time $0$, the center one 
is at time $1$ and the right one is from time $2$.
When there is an edge (or a direct communication) 
between agents $a$ and $b$ at a given time $t$, 
then $\chain{a}$ is true for agent $b$ at time $t$, 
for any time budget. 
Now observe that $d$ does not communicate with $c$ 
at time $0$, but does so at time $1$. 
This means that, for $\chain{c}$ to be true at time 
$0$ for agent $d$, we need a time budget of at 
least $1$. Since $c$ communicates with $b$ until 
time $2$, $\chain{b}$ is also true for $d$ at time 
$0$ with a time budget of $1$. 
Similarly, since $d$ and $e$ communicate at time 
$0$, there is a causal communication chain 
$e \to d \to c \to b$ that only requires a time 
budget of $1$. Observe, however, that the reverse 
communication chain is not causal because $d$ and 
$e$ communicate at time $0$, while $d$ and $c$ 
communicate at time $1$. 
Therefore, $e$ does not receive information from 
$c$ via $d$, which means that $\chain{e}$ is false 
for agent $c$ at time $0$ for any time budget 
smaller than $2$.
In addition, to accommodate online monitoring, 
the semantics must take into account that we may 
not have enough information about the timed graph, 
as samples are received when available. 
For example, imagine that we want to monitor 
whether there is a causal communication chain from 
$e$ to $a$ within a time budget of $2$ at time $0$. 
From the analysis above, we know that this property 
is true if we have access to the three samples 
above. 
Now assume that we have only received the first 
two samples, but not the third one. 
In that case, we have not found a causal chain, but 
we have not yet received enough information to 
exhaust the allocated time budget. 
Therefore, the semantics should indicate that the 
truth value is not yet determined.
\end{example}

To write such specifications, \mutgl relies on modalities
for both time and space constraints. Precisely, we
have three modalities:
$\F{T}\fvar$, the usual ``eventually'' from MITL;
$\diam{D}\fvar$, similar to $\F{}$ but for space,
requiring the existence of an agent at a distance in
the interval $D$; and
$H_h\fvar$, which fixes a time budget of $h$;
together with the usual Boolean connectives and atomic
propositions, which correspond to individual agents.
Given a timed graph $\tgraph$, the
semantics of such a formula $\phi$ is then, as described
above, a function of type
$\map{\semfun{\phi}{\tgraph}}{\preals\times\reals\times\agents}{\threeval}$,
where $\threeval$ is the three-valued ordered set $\bot \leq \undet \leq \top$,
consisting of Booleans extended with an indeterminate value.
The intention is that $\semwofp{\fvar}{\tgraph}{t}{h}{a}$ gives the truth
value of $\fvar$ for agent $a$ at time $t$ with time budget $h$.
In particular, for the three modalities:
\begin{itemize} 
	\item $\semwofp{\F{T}\:\fvar}{\tgraph}{t}{h}{a}= 
		\biglor_{t' \in T}\semwofp{\fvar}{\tgraph}{t+t'}{h-t'}{a}$ 
	\item $\semwofp{\diam{D}\:\fvar}{\tgraph}{t}{h}{a}= 
		\left\{ \begin{array}{lr} 
			\biglor_{b\in A}\semwofp{\fvar}{\tgraph}{t}{h}{b} \land \tgraph(t,a,b) \in D &\text{ if } h \geq 0\\ 
			\undet &\text{ otherwise } 
		\end{array} \right.$ 
	\item $\semwofp{H_{h'}\:\fvar}{\tgraph}{t}{h}{a}= 
		\left\{ \begin{array}{lr} 
			\bot &\text{ if } \semwofp{\fvar}{\tgraph}{t}{h'}{a} = \undet \text{ and } h \geq h'\\ 
			\semwofp{\fvar}{\tgraph}{t}{\min(h,h')}{a} &\text{ otherwise } 
		\end{array} \right.$ 
\end{itemize}
In particular, the time budget works as follows:
whenever we try to access the timed graph without having any
remaining budget, the result is indeterminate, which is
the case in $\diam{D}$. This indeterminate value is then resolved
to false when encountering a modality $H_{h'}$ with $h'$ small enough:
if we could not determine the truth value within the time
budget, then the formula is false.

Finally, in order to iterate these local time and space
properties, \mutgl contains a fixpoint operator $\mu X.\:\fvar$.
Its semantics is defined as a least fixpoint.
When a formula $\fvar$ has $X$ as a free variable, i.e.,
not bound by a $\mu$, the semantics $\semfun{\fvar}{\tgraph}$
can be seen as a function from
$\preals\times\reals\times\agents\to\threeval$ to itself.
This set is a complete lattice under the pointwise order.
Furthermore, as usual, when $\fvar$ is \emph{positive},
i.e., every branch of $\fvar$ ending in $X$ contains an even
number of negations, this function is monotone.
Then, by the Cousot--Cousot theorem~\cite{PC79},
the least fixpoint exists and can be computed by
iteratively applying the semantic function to the least
element of the lattice (namely, the constant function equal
to $\bot$).

\begin{example}[Horizons, continued]
For example, we can define:
\[\chain{a} = \mu X. \left(a \lor \F{[0,1]}\,\diam{[0,1]}\,X\right).\]
Starting from the constant function $f_0$ equal to $\bot$ and applying
the function $\semfun{a \lor \F{[0,1]}\diam{[0,1]}\:X}{\tgraph}$ once,
we obtain $f_1(t,h,\text{ag})= \top$ if $\text{ag}=a$, and $\bot$ otherwise.
After another iteration, we obtain that, for example, $f_2(0,h,\text{ag})$ is
$\top$ when $\text{ag}=a$, or when $\text{ag}=b$ and $h\geq 1$.
After enough iterations, we obtain that
$f_k(0,h,e)$ is $\top$ if $h\geq 2$ and $\undet$ otherwise.
This is actually the case for the least fixpoint as well, i.e.,
for $\semwofp{\chain{a}}{\tgraph}{0}{h}{e}$.
This means, in addition, that:
\[\semwofp{H_1\:\chain{a}}{\tgraph}{0}{2}{e} = \bot \qquad\qquad
\semwofp{H_2\:\chain{a}}{\tgraph}{0}{2}{e} = \top.\]
\end{example}

	\section{Space horizon for \mutgl}

	In this section, we extend the syntax and the semantics 
	of \mutgl to be able to write specifications such that
	``agent $a$ can send a message to agent $b$ with 
	a communication chain that costs less than $p$'' or 
	``agent $a$ can send a message that will propagate to 
	a distance more than $p$''.
	Those specifications are the Reach and Escape formulas 
	that are part of the syntax of STRel~\cite{LN22}.
	To follow the line of \mutgl, we will not have those 
	specifications natively in our extension because they 
	are global properties on the graph. Instead, we will 
	\emph{define} them thanks to our local modalities and the 
	usage of a fixpoint. It turns out that the only piece missing
	in the syntax is a \emph{space horizon}, similar to 
	the modality $H_h$ of vanilla \mutgl. This has two 
	advantages: 1) the modification of the logic is seemingly 
	minimal (although some non-trivial work has to be done 
	in the semantics) 2) this allows to defined Reach and 
	Escape in a causal and timed way. Indeed, for the latter,
	in STRel, those constructions are untimed in the sense 
	that they only look at the communication graph at a 
	fixed time.
	Concretely, 
	\begin{definition}[Agent Formulas]
		\emph{Agent formulas} are given by the following grammar:
		\[
			\fvar ::=  \top \spb p \spb \neg \fvar \spb \fvar \land \fvar \spb 
			\diam{D} \fvar \spb \F{T} \fvar \spb 
			X \spb \mu X. \fvar \spb H_h \phi  \spb 
			a \spb \exists a.\fvar 
			\spb \textcolor{red}{\spaceHorizon{s} \fvar},
		\]
		where the only addition to the syntax of \mutgl is  
		the modality $\spaceHorizon{s} \fvar$, the space horizon.
		In this grammar: $p$ is a predicate, $D$ and $T$ are closed intervals 
		of non-negative reals, $a$ is an agent variable, $X$ is a 
		fixpoint variable, and $h$ and $s$ are non-negative 
		reals.
	\end{definition}

	\begin{example}[Horizons, continued]
	We have seen that for 
	$\semwofp{\chain{a}}{\tgraph}{0}{h}{e}$ to be true, 
	we need a time budget of at least 2, so that we have 
	a causal communication chain from $e$ to $a$, 
	meaning essentially that $e$ satisfies $H_2\:\chain{a}$ but 
	not $H_1\:\chain{a}$ at time $0$. 
	Now observe that the unique causal communication 
	chain from $e$ to $a$ has length $4$. The intention 
	is then that $e$ satisfies $\spaceHorizon{s}\:H_h\:\chain{a}$
	at time $0$ if and only if $s\geq 4$ and $h\geq 2$.
	\end{example}

	The goal now is to formalize this intention in the 
	semantics. As for the time horizon, we need to keep track 
	explicitly of the space horizon in the semantics. 
	As such, the extended semantics will be a function 
	of type
	$\map{\semextfun{\fvar}{\sigma,\rho}}{\preals\times\reals\times\ereals\times\agents}{\threeval}$
	where $\sigma$ is a timed graph and $\rho$ is a valuation
	mapping every non-binded agent variable to an agent and every 
	non-binded fixpoint variable to a function of the same type as 
	the extended semantics. The arguments of those functions are as follows:
	1) the time $t\geq 0$ at which we check the satisfaction, 
	2) the time horizon $h$,
	3) the space horizon $s$, and
	4) the agent for which we are checking the satisfaction.
	Observe the crucial point that the space horizon 
	is an \emph{extended} real, meaning that it can 
	be infinite. The reason is we need to keep track of 
	the semantics when we have the full visibility of the 
	graph at a given time, which is not necessarily 
	the same as having the visibility with a large enough
	horizon as we will see later.
	Concretely, the extended semantics is defined 
	by induction on the formula as follows (we only 
	focus on the interesting cases)
	\begin{figure}[t]
	\begin{align*}
		\semExt{\diam{D}\: \fvar}{\sigma,\rho}{t, h, s,a,b} &=  \left\{ 
		\begin{array}{lr} 
			\top \; &\text{ if } \sigma(t,a,b) \in D\\ &\land \; \sigma(t,a,b) \leq s \\ &\land \; \semExt{\fvar}{\sigma,\rho}{t, h, s-\sigma(t,a,b),b} = \top\\ &\\
			\undet \; &\text{ if } \left(\sigma(t,a,b) > s \land \max(D) > s\right) ~\lor\\
			&\big(s \geq \sigma(t,a,b) \land \sigma(t,a,b) \in D \\ &\;\;\;\; \semExt{\fvar}{\sigma,\rho}{t, h, s-\sigma(t,a,b),b} = \undet\big)\\ &\\
			\bot &otherwise
		\end{array}
		\right.\\
		\semExt{\diam{D}\: \fvar}{\sigma,\rho}{t, h, s,a}&= 
		\left\{ 
		\begin{array}{lr} 
			\biglor_{b \in A} \semExt{\diam{D}\: \fvar}{\sigma,\rho}{t, h, s,a,b} &\text{ if } h \geq 0\\
			\undet &\text{ otherwise }
		\end{array}
		\right.\\
		\semExt{\spaceHorizon{s'}\:\fvar}{\sigma,\rho}{t, h, s, a}&= \left\{ 
		\begin{array}{lr} 
			\bot &\text{ if } \semExt{\fvar}{\sigma,\rho}{t,h,s',a} = \undet\\ 
				&\;\;\;\;\land~ s \geq s'\\
			\semExt{\fvar}{\sigma,\rho}{t, h, \min(s,s'), a} &\text{ otherwise }
		\end{array}
		\right. \\
		\semextfun{\mu X.\: \fvar}{\sigma,\rho}&= 
			\lfp(f \mapsto \semExt{\fvar}{\sigma,\rho [X\rightarrow f]}{\sigma,-,-,-,-})
	\end{align*}
	\caption{Extended semantics.}
	\label{fig:ext-sem}
	\end{figure}
	\begin{definition}[Extended Semantics]
	We define the extended semantics as in Figure~\ref{fig:ext-sem}
	where the least fixpoint is taken in the 
	complete lattice of functions of type 
	$\preals\times\reals\times\ereals\times\agents\to\threeval$ 
	with the pointwise order. The rest is defined as 
	in vanilla \mutgl.
	\end{definition}

	Observe that the semantics of the space-horizon 
	modality is similar to the time-horizon. Most of 
	the work is done in the semantics of the spatial 
	modality $\Diamond$. The intention of $\Diamond_{D}\:\fvar$ 
	is to check if there is another agent satisfying $\fvar$, 
	for which the distance is within the interval $D$.
	In vanilla \mutgl, that is where indeterminate 
	can appear: this modality is where we actually have 
	to have access to the graph, and if we do not have 
	time budget anymore, we cannot access the graph, 
	and we cannot know if there exists such an agent.
	In our extension, there is another case where this 
	could happen because of the space budget. 
	An edge between two agents $a$ and $b$ is valid for this 
	modality if its weight is within the interval $D$ 
	\emph{and within the space budget}. In that case, 
	if furthermore $b$ satisfies $\fvar$ with the discounted 
	space horizon, then $b$ is a witness that $a$ 
	satisfies $\Diamond_{D}\:\fvar$. If no $b$ satisfy 
	these requirements, but if the space horizon $s$
	is smaller than the maximal value of $D$ (meaning 
	that the current space budget does not allow to cover 
	the whole interval $D$) and if the weight of the 
	transition between $a$ and $b$ is bigger than $s$ 
	(i.e., the current space horizon does not allow 
	to see this transition), then $b$ may be a witness
	but we cannot decide yet.

	\begin{example}[Horizons, continued]
		\label{ex:last-horizon}
	
	With this new semantic, the expression $chain_a$ evaluates as follows: at the base step of the fixpoint construction, the formula holds for agent $a$ only, and leaves the other agents either false or indeterminate due to lack of horizon: 
	\[
	f_1(t,h,s,ag) = 
	\left \{
	\begin{array}{ll}
	\top &\text{ if } ag = a \\
	\undet &\text{ if } ag \neq b \text{ and } s < 1 \lor h = 0 \\
	\bot &\text{ otherwise }\\
	\end{array}
	\right.
	\]
	
	As the iteration count increases, the recursive application of the semantic allows other agents to be valid dependently on the budget allowed for space horizon. For example, agent $b$ will satisfy the constraints as long as the time horizon permits to see the third graph and $s \geq 1$.
	After completing the full fixpoint computation, we end up with a semantic that confirms that:
	\[
	\semextfun{\chain{a}}{\sigma,\rho}(0,2,4,e)=\top \qquad\qquad
	\semextfun{\chain{a}}{\sigma,\rho}(0,2,s,e)=\undet
	\]
	for $s<4$. This gives the following intended result: 
	\[
	\semextfun{\spaceHorizon{3}\:H_2\:\chain{a}}{\sigma,\rho}(0,2,4,e)=\bot \qquad\qquad \semextfun{\spaceHorizon{4}\:H_2\:\chain{a}}{\sigma,\rho}(0,2,4,e)=\top
	\]
	
	Also note that in this example, a naive computation 
	of the fixpoint takes a infinite amount of time. 
	See what is happening at time 0 and horizon 2: 
	$f_n(1,1,s)$ keeps all agents either $\top$ or 
	$\undet$ if $s<n$. But when $s\geq n$, all 
	indeterminate agents become $\bot$, hence 
	determinizing all agent as if there were no 
	restriction. Future iterations propagate these 
	indeterminate agents, leading to an infinite computation 
	if we don't take that into account. 
	In this case, the fixpoint result of 
	$\semextfun{\chain{a}}{\sigma,\rho}(1,1,s)$ would 
	give the agent $e$ as indeterminate  when 
	$s \in [5,\infty)$, and $\bot$ at $s = \infty$.
	\end{example}

	As we saw, 
	the semantics is not 
	continuous: we cannot recover 
	the semantics for $s=+\infty$ as the limit 
	of a sequence of values of the semantics for 
	finite values of the space horizon.
	However, it is monotone in the space horizon:

	\begin{lemma}
		\label{lem:monotony}
		For all agent formula \fvar, monotone context $\rho$, timed graph $\sigma$, $t \in \preals$, $h\in\reals$, $s \leq s'\in\ereals$, and $a\in\agents$: 
		\[ \semExt{\fvar}{\sigma,\rho}{t,h,s,a} \preceq \semExt{\fvar}{\sigma,\rho}{t,h,s',a}
		\]
		where $\preceq$ is the information partial order $\undet \preceq \top,\bot$
		and
		a context is monotone if for all fixpoint variable $X$
		$\rho(X)$ is monotone in the space horizon as above.
	\end{lemma}

	The value at infinity, which means that we do 
	not put any restriction on the space horizon 
	we can see, actually correspond to the 
	vanilla semantics of \mutgl:
	
	\begin{lemma}
		\label{lem:vanilla-infinity}
		For formula \fvar\! of \mutgl, monotone context $\rho$, signal $\sigma$, $t,h \in \mathbb{R}_+$: 
		\[
		\sem{\fvar}{\rho}{\sigma, t, h} = \semExt{\fvar}{\rho}{\sigma, t, h,\infty}
		\]
		
	\end{lemma}

	\begin{example}[Causal Reach and Escape]
		\label{ex:causal-reach}
	
	The Reach and Escape modalities are fundamental components of the STREL logic, 
	showing their importance in spatio-temporal analysis. These operators allow the verification of connectivity and propagation of an agent network at a  fixed time. 
	For any propositions $\fvar_1$ and $\fvar_2$, Reach modality, denoted as $\fvar_1 \mathcal{R}_s \fvar_2$, asserts the existence of a communication chain connecting a source and a destination with constrains on the chain's agents and total weight. 
	Formally, the property holds if there exists a communication chain $ag_1 \xrightarrow{d_1} ag_2 \cdots \xrightarrow{d_{n-1}} ag_n$ where $\fvar_1$ holds for all intermediate agents in the chain $ag_{\leq n-1}$ and $\fvar_2$ is true at the final agent $ag_n$ and such that $(\sum_{i \in [0,n-1]}d_i) \leq s $. 
	In \mutgl, it can be written as:
	\[
	\fvar_1 \mathcal{R}_s \fvar_2 = \exists a. \:a \leq \spaceHorizon{s}\:\mu X. (\fvar_2 \lor (\fvar_1 \land \diam{\leq s\;} X) )
	\]
	Observe that we introduce the modality $\phi\leq\phi'$, which is already defined in standard \mutgl. It compares the values of each agents at $\phi$ and $\phi'$ using the pointwise order. We use it the same way in our extension:
	\[ \semExt{\fvar\leq\fvar'}{\sigma,\rho}{t,h,s} = \bigland_{a \in A}\big( \semExt{\fvar}{\sigma,\rho}{t,h,s,a} \to \semExt{\fvar'}{\sigma,\rho}{t,h,s,a} \big)
	\]

	Escape property is written $ \mathcal{E}\fvar$ and 
	checks if the property $\fvar$ can be propagated enough 
	in space. Formally, it means that there exists a 
	destination agent $ag$ where a communication chain 
	in which $\fvar$ holds can be created. 
	In \mutgl:
	\[
	\mathcal{E}_s\fvar =  \exists b.\; a \leq  (\mu X.( \fvar \land (b \lor \diam{\leq s\;} X) ) \land \neg (\spaceHorizon{s\;}\mu X.(b \lor \diam{\leq s\;} X ))) 
	\]

	However, those standard formulations do not take into account the temporal dimension of message transmission. The extension proposed offers this possibility leading to these new formalization of reachability and escape:
	\[
	\fvar_1 \mathcal{R}_s \fvar_2 = a \leq \spaceHorizon{m}\;H_t\;\mu X. (\varphi_2 \lor (\varphi_1 \land \F{[0,t]}\diam{\leq m\;} X) )
	\]
	\[
	\mathcal{E}_s\fvar = \exists b.\; a \leq  (H_t \;\mu X.( \fvar \land (b \lor \F{[0,t]}\diam{} X) ) \land \neg (H_t \spaceHorizon{m\;}\mu X.(b \lor \F{[0,t]}\diam{\leq m\;} X ))) 
	\]

	\end{example}
	
	\begin{corollary}
		Extended $\mu$TGL is strictly more expressive than $\mu$-TGL.
	\end{corollary}

	\section{Computing the Extended Semantics}

	Computing the semantics of modal logics is always tricky 
	because we may need infinite amount of information 
	to be able to compute a local part of the semantics.
	For temporal logics, this translates in the following:
	to compute the semantics at a given time, we may need 
	the semantics of subformulas for a unbounded 
	amount of time. This can be overcome 
	by ensuring that the intervals in the modalities 
	are bounded. The fixpoint operator of \mutgl 
	brings additional complications: even if 
	the intervals are bounded, the fixpoint may not 
	converge in a finite number of iterations (which is 
	a problem in itself) and we may have to apply the 
	modality arbitrary many times, and so, we 
	may need unbounded amount of information again.
	To overcome this problem, the monitoring algorithm 
	of \mutgl is only dealing with formulas that only 
	need a finite horizon of time, meaning that 
	any fixpoint which contains a $\F{T}$ modality, must 
	be guarded by a time horizon modality $H_h$. When 
	the timed graph is assumed piecewise-constant 
	with finitely many changes of values within a 
	time window bounded by this horizon, we can then 
	compute the semantics, and in particular, the 
	fixpoints terminate after finitely many iterations.
	Concretely, after computing this horizon $h_\phi$, 
	what is computed is a restriction of the 
	semantics, namely
	$\map{\semfun{\fvar}{\sigma,\rho}}{\preals\times[0,h_\phi]\times\agents}{\threeval}$.

	In our extension, we would like something similar, 
	however, we do not want to make additional 
	restrictions on the set of formulas we can 
	monitor: asking that fixpoint are guarded 
	by space-horizon modalities means we could not 
	monitor formulas 
	from vanilla \mutgl, seen as formulas of this 
	extensions.
	It turns out such restrictions are not needed:
	1) either we are under the scope of a space-horizon, 
	in which case we can bound the amount 
	of information we need, similarly to the time horizon 
	2) or we are not, and what we need is the value at 
	infinity, which can be computed similarly to 
	the semantics of vanilla \mutgl.

	\subsection{Bounding the Information}

	To make this statement concrete, we introduce a new 
	semantics, which defines the needed 
	(bounded) information of the semantics by induction:
	\begin{definition}[Concrete Semantics]
	Fix $\fvar$ and define $s_{\phi}$ to be the 
	maximal $s\in\preals$ seen as $\spaceHorizon{s}$ in 
	$\fvar$. We define the concrete semantics as 
	a function of type:
	\[
	\rsemextfunbis{\fvar}{\sigma,\rho}\,\colon\,\preals\times[0,h_\phi]\times\left([0,s_{\fvar}]\cup\{+\infty\}\right)\times\agents\to\threeval
	\]
	by induction on the formula $\fvar$, using the same 
	definition as for $\semextfun{\fvar}{\sigma,\rho}$, 
	except that least fixpoints and context $\rho(X)$ 
	are over functions of type 
	$\preals\times[0,h_\phi]\times\left([0,s_{\fvar}]\cup\{+\infty\}\right)\times\agents\to\threeval$.
	\end{definition}

	\begin{lemma}
		\label{lem:restriction}
		For every context $\rho$ with $\rho(X)\,\colon\,\preals\times\reals\times\ereals\times\agents\to\threeval$, 
		we have:
			\[\rsemextfunbis{\fvar}{\sigma,\overline{\rho}} = \overline{\semextfun{\fvar}{\sigma,\rho}}\]
		where $\overline{f}$ is the restriction to 
		$\preals\times[0,h_\phi]\times\left([0,s_{\fvar}]\cup\{+\infty\}\right)\times\agents\to\threeval$
		of a function or a context.
	\end{lemma}
	The intuition of this lemma is that computing the 
	whole semantics and then restricting to the part we really need 
	is the same as doing the computation with restricting
	all the intermediate results to the part we need.
	This means that one can always bound the amount of 
	information that is needed to compute the needed part 
	of the semantics.

	\subsection{Data Structure}
	
	Still, bounded information does not mean that the 
	computation -- particularly the fixpoints -- will
	terminate. 
	For simplicity, and as it is usually done for monitoring 
	algorithms of temporal logics, we will assume that 
	the timed graph is given by samples at runtime, and that 
	the timed graph is constant between two samples.
	This introduces imprecisions in the analysis, but 
	remains a reasonable assumption if the signal is 
	sufficiently sampled and the dynamics not too wild.
	In this context (and modulo some mild non-Zeno conditions)
	we will prove that computing samples of the 
	semantics terminates, and that we can then perform 
	monitoring.
	Concretely, 
	the monitoring algorithm works with timed graphs 
	represented by a sequence of samples of the form 
	$
	x = ((t_0,\sigma_0),\ldots,(t_n,\sigma_n),\ldots)$,
	where $t_0 < t_1 < \ldots$ is an increasing sequence of 
	$\preals$ and 
	where $\sigma_i$ is a weighted graph whose vertices are agents $\agents$ 
	which denotes the state of the timed graphs at times $[t_i,t_i+1)$. 
	Accordingly, we process the signal using the 
	piecewise constant functions defined by 
	$\sigma(t,a,b) = \sigma_i(a,b)$ if $t \in [t_i,t_{i+1})$.
	We assume that the sequence $(t_n)_{n\in\mathbb{N}}$
	does not converge to a finite value.
	This implies that 1) the sequence diverges to infinity and 
	2) in any given bounded interval of $\preals$, there 
	will be only a finite number of $t_n$.
	This data structure is used both for the input timed 
	graph, but also for the semantics to be computed.
	Indeed, the semantics depends on three real values: 
	the time $t$, the time horizon $h$, and the space 
	horizon $s$. 
	That is, up to currying, the semantics is of type:
	\[\semextfun{\fvar}{\sigma,\rho}\:\colon\:
		\preals\to\reals\to\ereals\to\agents\to\threeval.\]
	Therefore, the result will be a signal of signals of signals of uncertain sets of the form:
	for $i, j, k\in\mathbb{N}$,
	$w = \{ \;(t_i,
	\{\; (h_{i,j},
	\{\;(s_{i,j,k}, \text{Uset}_{i,j,k})\;\})
	\;\} )\; \}$.
	We call such a data structure a \emph{discrete signal}. 
	The last model to represent is the uncertain sets: using the same data structure as \cite{ER22}, the  uncertain sets Uset : $ A \rightarrow 3$ are represented by a pair $[U,V]\;( U,V \in \mathcal{P}(A))$ such that 
	\[
	[U,V](a) = 
	\left\{
	\begin{array}{ll}
	\top &\text{ if } a \in U\\
	\undet &\text{ if } a \notin U \land a \in V\\
	\bot &\text{ if } a \notin V\\
	\end{array}
	\right.
	\]
	Finally, given such a data structure, we obtain 
	a continuous-time function:
	\[
	\bar{w}(t,h,s,a) = \left\{
	\begin{array}{ll}
		\top &\text{ if } a \in U_{i,j,k}\\
		\undet &\text{ if } a \notin U_{i,j,k} \land a \in V_{i,j,k}\\
		\bot &\text{ if } a \notin V_{i,j,k}\\
	\end{array}
	\right.
	\]
	whenever $t_i \leq t < t_{i+1} \text{, } h_{i,j} \leq h-(t-t_i) < h_{i,j+1} \text{ and } s_{i,j,k} \leq s < s_{i,j,k+1}$.

	\subsection{Termination of the Computation}

	What remains to be done is to prove that the 
	computation of the semantics for a bounded 
	amount of time $T$, 
	that is, the function 
	$\rsemextfunbis{\fvar}{\sigma,\rho}$
	restricted to $[0,T]\times[0,h_\fvar]\times\left([0,s_\fvar]\cup\{+\infty\}\right)\times\agents$
	can be computed.
	In particular, this requires to prove that all the 
	fixpoints involved in the computation terminate in a 
	finite number of iterations.
	To prove this, we construct finite sets 
	$\mathcal{T}$,
	$\mathcal{H}$, 
	and $\mathcal{S}$ 
	of reals, that only depend on $\sigma(t',-,-)$
	with $t'\in[0,T+h_\phi]$ and 
	$\fvar$. We would then prove by induction that 
	the computation terminates and that there 
	is a discrete signal
	\begin{equation}
	w =
	\{\; (t,
		\{\; (h,
			\{\;(s, 
				\text{Uset}_{t,h,s}
			)\;\}
		)\;\}  
	)\;\}
	\label{eq:discrete-signal}
	\end{equation}
	such that $\rsemextfunbis{\fvar}{\sigma,\rho}=\overline{w}$,
	$t\in\mathcal{T}$,
	$h\in\mathcal{H}$, and 
	$s\in\mathcal{S}$.
	Now, fix $T$. By assumption, there is a finite 
	number of $0=t_0, \ldots, t_{n}$ in $[0,T+h_\phi]$.
	We build $\mathcal{T}$, 
	$\mathcal{H}$, and $\mathcal{S}$ as 
	follows. They are the smallest sets satisfying:
	\begin{itemize}
		\item 1) $t_i\in\mathcal{T}$, for 
			$i \leq n$ and 2) 
			if $t\in\mathcal{T}$, then $t+a$, 
				$t+b\in\mathcal{T}$, for every 
				$[a,b]$ appearing in $\F{[a,b]}$
		\item 1) $0\in\mathcal{H}$, 
			2) $h_\phi\in\mathcal{H}$, 
			3) $h\in\mathcal{H}$ for every 
				$h$ appearing as $H_h$, and
			4) if $h\in\mathcal{H}$, then $h-a$, 
				$h-b\in\mathcal{H}$, for every 
				$[a,b]$ appearing in $\F{[a,b]}$
		\item 1) $0\in\mathcal{S}$,
			2) $s\in\mathcal{S}$ for every 
				$s$ appearing as $\spaceHorizon{s}$
			3) $b\in\mathcal{S}$ for every 
				$b$ appearing as $\Diamond_{[a,b]}$, and
			4) for $i\leq n$, $a,b\in\agents$, and 
				$s\in\mathcal{S}$, then
				$s-\sigma(t_i,a,b)\in\mathcal{S}$.
	\end{itemize}
	intersected by $[0,T+h_\fvar]$, $[0,h_\fvar]$, and $[0,s_\fvar]$
	respectively. Finally, we add $+\infty$ to $\mathcal{S}$.
	Those sets are finite as basically 
	finite unions of finitely generated sub-monoids 
	of $(\mathbb{R},+)$ intersected with bounded intervals.
	
	\begin{lemma}
		\label{lem:termination}
	For every subformula $\psi$ of $\phi$ and 
	every context $\rho$ where every $\rho(X)$
	is obtained as $\overline{w}$ for some $w$
	as in \eqref{eq:discrete-signal}, 
	then $\rsemextfunbis{\psi}{\sigma,\rho}$
	is also such a function.
	Furthermore, if $\psi\equiv\mu X. \psi'$, then 
	the fixpoint defining its semantics terminates.
	\end{lemma}

	\subsection{Computing semantics}
	
	As in \cite{ER22}, in order to manipulate the data structure, we need to implement a way to zip two signals together, that is construct a procedure $Zip$ that takes two signals $x$ and $y$, and produces a composite signal $x\times y $ satisfying: 
	$\overline{x\times y}(t) = (\bar{x}(t),\bar{y}(t)).$
	In particular, the union of two signals requires zipping corresponding each layer of signals of the semantic, and then merging the associated uncertain sets: $[U^x_{i,j,k},V^x_{i,j,k}],[U^y_{i,j,k},V^y_{i,j,k}]$ to get $[U^x_{i,j,k} \cup U^y_{i,j,k},V^x_{i,j,k} \cup V^y_{i,j,k}]$.
	For most operators, the semantic computing is not 
	too different from the constructions in \cite{ER22}. 
	One operator interesting to explain is the diamond. 
	The semantics of the diamond operator $\diam{D}\fvar$ is computed 
	according to Algorithm~\ref{alg:diam}. 
	For each $t_i$ and $h_{i,j}$ constituting the semantic $w_\fvar$ of \fvar, 
	we construct subsets $w_{a,b}$ of the final semantics, parameterized by each 
	weight $d = \tgraph(t_i,a,b)$.
		If $d \notin D$, the value $w_{a,b}(t_i,h_{i,j},0,b)$ will be $\undet$ (l.6, left). 
		It remains unchanged until the space horizon reaches $\min(d,\max(D))$, at which point 
		the values for agent $b$ is confirmed to be $\bot$ (l.6, right).
		If $d \in D$, we initialize $w_{a,b}(t_i,h_{i,j},0,b)$ at $\undet$ (l.8). 
		The semantics at space horizon $s$ depends on the value of $w_\fvar(t_i,t_{i,j},s-d)$. 
		By leveraging the monotonicity of space horizon, it suffices to find the first space horizon value in which agent $b$ becomes 
		determinized to compute its semantics (ll.9--14). 
		If the agent remains indeterminate, we only add the $\undet$ value at space horizon $+\infty$ (l.16). 
	The result is obtained by zipping and unifying (l.17).
	
	\algtext*{EndIf}
	\algtext*{EndFor}
	\begin{algorithm}[tb]
		\caption{Diamond semantic}\label{alg:diam}
		\begin{algorithmic}[1]
			\Require a formula $\diam{[d_1,d_2]}\fvar$, a signal $\sigma$
			\Ensure semantic $\sem{\diam{[d_1,d_2]}\fvar}{\rho}{\sigma,t,h,s}$ for all $t,h,s$
			\State $res$ $\gets$ empty semantic; $sem_\fvar \gets$ semantic of $\fvar$
			\ForAll {t,h in the data structure of $sem_\fvar$}
			\ForAll {edge $a \xrightarrow{d} b$}
			\State $\text{tmp} \gets \bot$
			\If {$d \not \in [d_1,d_2]$}
			\State $\text{tmp}(t,h,0) \gets [\emptyset,\{b\}]$; $\text{tmp}(t,h,\min(d,d_2)) \gets [\emptyset,\emptyset]$
			\EndIf
			\If {$d \in [d_1,d_2]$}
			\State $\text{tmp }(t,h,0) \gets [\emptyset,\{b\}]$
			\If {$\exists s, a \in sem_\fvar(t,h,s)[U]$}
			\State $s_{min} \gets \min \;\{s, a \in sem_\fvar(t,h,s)[U] \}$
			\State $\text{tmp} (t,h,s_{min}+d) \gets [\{b\},\{b\}]$; $\text{tmp} (t,h,+\infty) \gets [\{b\},\{b\}]$
			\ElsIf {$\exists s, a \notin sem_\fvar(t,h,s)[V]$}
			\State $s_{min} \gets \min \;\{s, a \notin sem_\fvar(t,h,s)[V] \}$
			\State $\text{tmp}(t,h,s_{min}+d) \gets [\emptyset,\emptyset]$; $\text{tmp} (t,h,+\infty) \gets [\emptyset,\emptyset]$
			\Else \State $(\text{tmp} (t,h,+\infty) \gets [\emptyset,\{b\}]$
			\EndIf
			\EndIf
			\State $res(t,h) = Union(Zip(res(t,h), \text{tmp}))$
			\EndFor
			\EndFor
		\end{algorithmic}
	\end{algorithm}

	\section{Case Study}

	To illustrate the feasibility of our monitoring 
	algorithm, we analyze specifications related to 
	bounding the causal communication graph of a 
	collection of agents collaborating to visit 
	specified positions in a 2D-area. To dynamically 
	allocate which agent should visit which position,
	they execute a Consensus-Based Bundle Algorithm 
	(CBBA) as in \cite{CH09}. Roughly, every 
	connected component of the communication graph 
	(or bundle) run a bidding protocol where 
	each agent bids on specific positions, 
	depending on how difficult or far it is for 
	them to reach the position. After some rounds 
	of communications, the agent with the best bid on 
	a particular position can start going towards this 
	position to visit it.
	Crucially, the number of rounds of communications 
	needed for a bundle to reach a consensus depends 
	on the diameter of the communication graph.
	This is then a suitable problem for monitoring 
	that the diameter of the causal communication 
	graph be bounded to ensure the protocol runs 
	efficiently.

	\subsection{Simulation and Monitoring Framework}

	To obtain traces of simulation on which to 
	run our monitoring algorithm, we will rely 
	on the simulator~\cite{JI26}, based on~\cite{JI24}, and 
	implementing the CBBA protocol of~\cite{CH09} in Python.
	In this simulation framework, we can 
	choose the number of agents, the number of 
	positions to visit, the size of the area, 
	and the communication radius, and the various 
	positions are then drawn randomly. It is then 
	easy to extract the position of each agent 
	during the simulation, compiled in a CSV file.
	The code of our monitoring algorithm, together with 
	some experiment setups to interact with the simulator
	is available on github~\cite{BDM26}.
	Our code is in Ocaml 5.3.0, compiled with Dune, 
	and ran on a 
	Apple M4 Pro with 24GB RAM.

	\subsection{Monitoring the Diameter of the Causal Communication Graph}

	We ran the simulator with the following parameters:
	10 agents, visiting 100 positions, in a $1400\times 1000$ 
	area, with a $500$ communication radius (the specific unit does not matter).
	This produced a trace with $\sim\!12$K time steps.
\begin{figure}[t]
    \centering
    \includegraphics[width=0.40\textwidth]{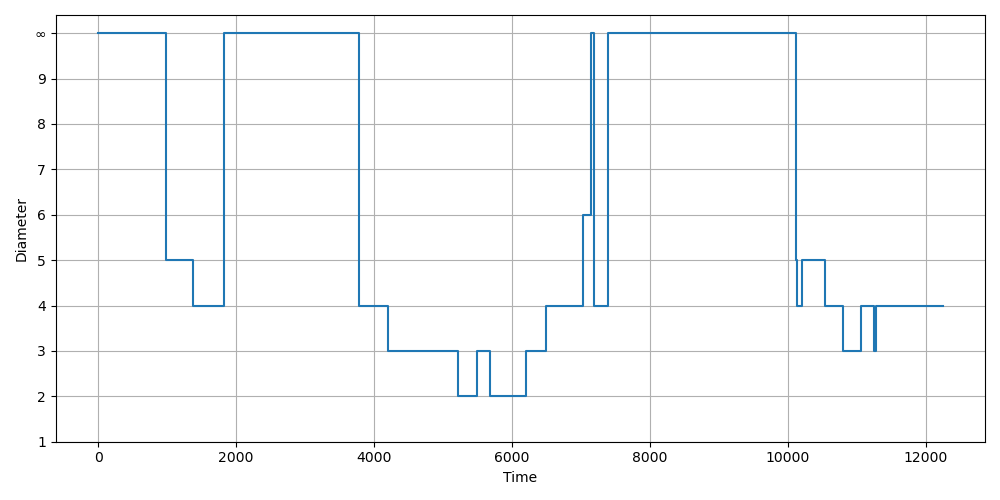}
    \qquad\qquad
    \includegraphics[width=0.40\textwidth]{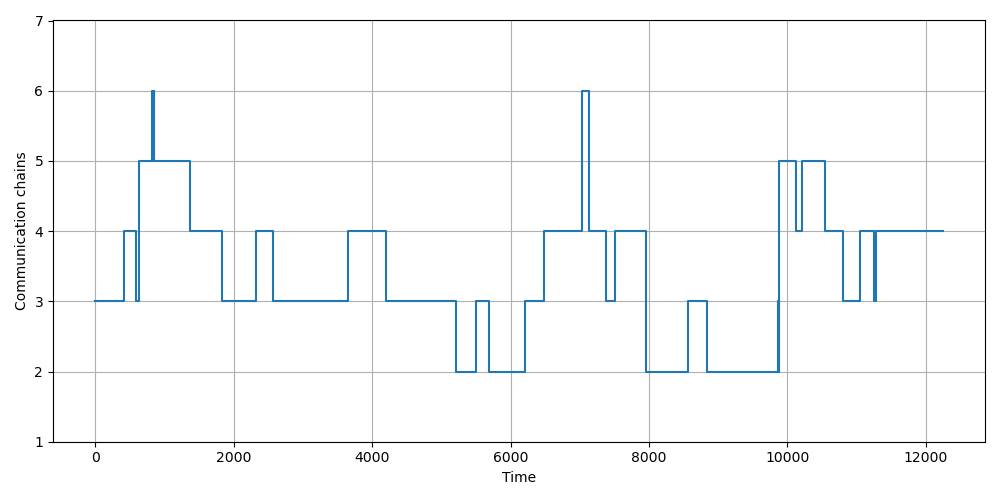}
    \caption{Bounds on the diameter (left) and the communication 
	chains (right) of the causal communication graph over time.}
    \label{fig:exp1}
\end{figure}
	On this specific trace, we monitored a specification
	describing that the causal communication graph within 
	10 time units has a diameter bounded by a 
	constant $d$, namely:
	\begin{equation}
		\label{spec:com-graph}
	\forall a.\, \top \leq 
	\spaceHorizon{d}\, H_{10}\,
	\mu X.\, \left(
		a \lor \F{\leq 10}\,\diam{\leq 1}\,X
	\right),
	\end{equation}
	very similar to the causal reach specification 
	of Example~\ref{ex:causal-reach}.
	We ran the monitoring algorithm  
	for all possible values of $d=\{1,...,9\}$, 
	for an average of $267$s per run (or $0.022$s per 
	time step), without significant difference in time 
	between runs.
	This allows to draw the Figure~\ref{fig:exp1} (left), showing 
	the connectivity of the causal communication graph.
	For example, this shows that it is fully connected 
	only on $[988,1830]\cup[3786,7137]\cup[7186,7385]\cup[10114,12243]$, 
	in which case the diameter is at most $6$.
	The diameter goes down to $2$ on $[5210,5498]\cup[5686,6201]$, 
	but never down to $1$, meaning it is never a clique.
	Another interesting property to monitor is to bound 
	the chains of communications of agents that have at 
	least one chain of communication to an agent $a$, 
	namely
	\begin{equation}
		\label{spec:com-graph-2}
	\forall a.\, 
	H_{10}\,
	\mu X.\, \left(
		a \lor \F{\leq 10}\,\diam{\leq 1}\,X
	\right) 
	\leq 
	\spaceHorizon{d}\, H_{10}\,
	\mu X.\, \left(
		a \lor \F{\leq 10}\,\diam{\leq 1}\,X
	\right),
	\end{equation}
	This specification is slightly stronger than the 
	fact that each bundle is bounded. The Figure~\ref{fig:exp1} (right)
	gives more information as for when the graph is not 
	fully connected: we recover some information on bounds 
	of communication chains.
	Observe that, for simplicity, we monitored one 
	specification for each value of the bound, and in both 
	cases of \eqref{spec:com-graph} and 
	\eqref{spec:com-graph-2}, while all the necessary information 
	is already available in the computation of the 
	semantics of the formulas
	$\spaceHorizon{9}\, H_{10}\,
	\mu X.\, \left(
		a \lor \F{\leq 10}\,\diam{\leq 1}\,X
	\right)$, for every agent $a$.

	\subsection{Scalabity w.r.t. the Number of Agents}

	We ran our monitoring algorithm on traces where the 
	parameters are: the number of agents $n$ is in $[3,5,10,20,30,40,50]$, 
	the number of positions is $m = 3n$, the area is $500\times 500$, 
	and the communication radius is $100$.
\begin{figure}[t]
    \centering
    \includegraphics[width=0.45\textwidth]{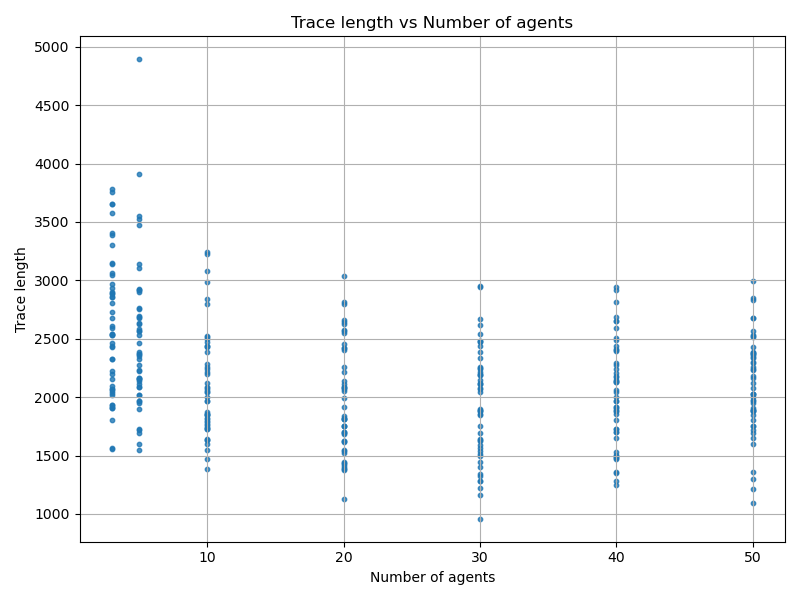}
    \hfill
    \includegraphics[width=0.45\textwidth]{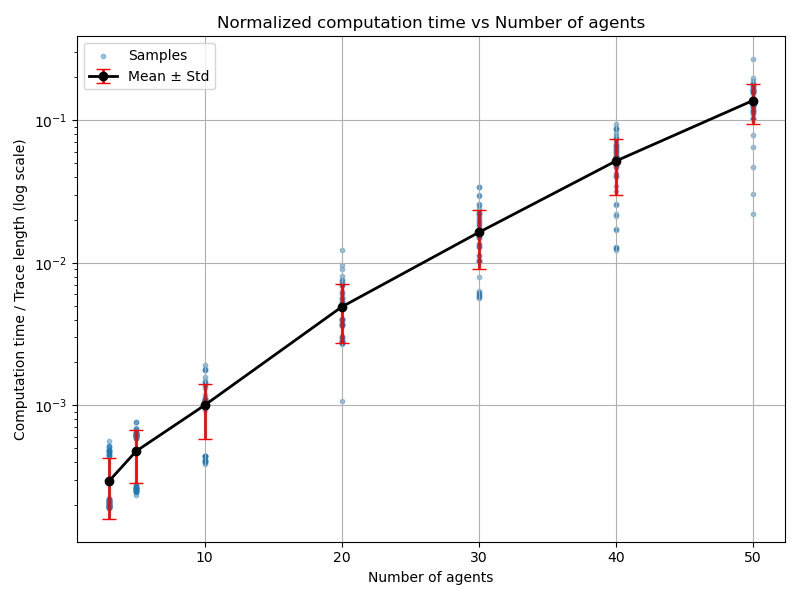}
    \caption{Length of traces and statistics on normalized times of computation (actual time divided by the trace length).}
    \label{fig:exp2}
\end{figure}
	This set up allowed us to have traces of similar lengths,
	even with different number of agents, so that 
	we can focus on the scalability w.r.t. this number.
	We monitored the specification:
	\begin{equation}
		\label{spec:com-tree}
	\top \leq 
	\spaceHorizon{n/2}\, H_{10}\,
	\mu X.\, \left(
		a \lor \F{\leq 10}\,\diam{\leq 1}\,X
	\right),
	\end{equation}
	that is, we only check that the causal communication
	chains towards one specific agent $a$ has a length 
	bounded by $n/2$.
	To recover the specification \eqref{spec:com-graph}, 
	we would have to check this for every agent,
	adding roughly a linear factor in the number of 
	agents.
	This specification is relevant as it would be enough 
	in a decentralized monitoring algorithm, 
	and is more realistic in the sense that this 
	would not require a central agent to know the 
	whole situation of the communication, while each agent can 
	monitor the specification \eqref{spec:com-tree} based only 
	on the communication it receives.
	We ran $50$ times for each number of agents, and 
	compile the results in Figure~\ref{fig:exp2}.
	In this setup, the traces are $\sim\!2K$ time step long.
	The monitoring scales nicely up to 30 
	agents: it takes an average of $17$ms per time step. 
	For 50 agents, the monitoring of 
	the specification \eqref{spec:com-tree} already takes 
	$142$ms per time step in average. This is fine 
	for physical systems with low velocity, 
	for which we would not need to monitor the 
	specification at high frequency, while 
	the additional factor necessary to monitor 
	the specification \eqref{spec:com-graph} would make it 
	impractical. 

	\section{Conclusion}

	We have provided an offline monitoring 
	algorithm with allows to monitor information about 
	the diameter of causal communication graphs, which 
	was not possible with previous spatio-temporal 
	logics. The next step is to make this algorithm 
	online, which will be allowed by the careful 
	semantics crafted in the present paper. 
	To do so, we would like however to integrate it 
	in a more general framework, and particularly in 
	a control loop, as a shielding. To go further in 
	that direction, we would like to develop a 
	quantitative semantics of \mutgl, allowing 
	to better discriminate control prediction and help 
	the control synthesis.
	
	\ackname

	This work is partially supported by the Agence
	de l'Innovation de Défense (AID) via the Centre Interdisciplinaire Mers et Océan (CIMO) 
	project 2025 CHRoMM.
	J.D. is partially funded by the Academic and Research Chaire ``Architecture des Systèmes Complexes'' 
	Dassault Aviation, Naval Group, Dassault Systèmes, KNDS France, Agence de l'Innovation de Défense, 
	Institut Polytechnique de Paris.

	\newpage

	\bibliographystyle{splncs04}
	\bibliography{ext_mu}

	\newpage
	\appendix

	\section{Full Semantics of Our Extension}

	\begin{align*}
		\semExt{\top}{\rho}{\sigma, t, h, s}&= \top\\
		\sem{p}{\rho}{\sigma, t, h, s}&=
		\left\{ 
		\begin{array}{ll} 
			\biglor_{b\in A}\big( p(\sigma(a,t)) &\text{ if } h \geq 0\\
			\undet &\text{ otherwise }
		\end{array}
		\right. \\
		\semExt{\neg \fvar}{\rho}{\sigma, t, h, s}&= \neg\sem{\fvar}{\rho}{\sigma, t, h, s}\\
		\semExt{\fvar \land \psi}{\rho}{\sigma, t, h, s}&=\sem{\fvar}{\rho}{\sigma, t, h, s} \land \sem{\psi}{\rho}{\sigma, t, h, s} \\
		\semExt{a}{\rho}{\sigma, t, h, s}&= \{\rho(a)\}\\
		\semExt{\diam{D}\: \fvar}{\rho}{\sigma, t, h, s}(a,b) &=  \left\{ 
		\begin{array}{ll} 
			\top \; &\text{ if } d(\sigma(a,t),\sigma(b,t)) \in D\\ &\land \; d(\sigma(a,t),\sigma(b,t)) \leq s \\ &\land \; \semExt{\fvar}{\rho}{\sigma, t, h, s-d(\sigma(a,t),\sigma(b,t))}(b) = \top\\ &\\
			\undet \; &\text{ if } d(\sigma(a,t),\sigma(b,t)) > s \land \max(D) > p\\
			&\lor p \geq d(\sigma(a,t),\sigma(b,t)) \land \semExt{\fvar}{\rho}{\sigma, t, h, s-d(\sigma(a,t),\sigma(b,t))}(b) = \undet\\ &\\
			\bot &\text{ if } p \geq \max(D) \land d(\sigma(a,t),\sigma(b,t)) > \max(D) \\
			&\lor \big(p \geq \min(D) \land \max \geq d(\sigma(a,t),\sigma(b,t)) \\ &\;\;\;\;\land \semExt{\fvar}{\rho}{\sigma, t, h, s-d(\sigma(a,t),\sigma(b,t))}(b) = \top \big)
		\end{array}
		\right.\\
		\semExt{\diam{D}\: \fvar}{\rho}{\sigma, t, h, s}(a)&= \biglor_{b \in A} \semExt{\diam{D}\: \fvar}{\rho}{\sigma, t, h, s}(a,b) \\
		\semExt{\F{T}\: \fvar}{\rho}{\sigma, t, h, s}(a)&= \biglor_{t' \in T}\sem{\fvar}{\rho}{\sigma, t+t', h-t', s}\\
		\semExt{\exists a.\:\fvar}{\rho}{\sigma, t, h, s}&= \biglor_{b\in A} \sem{\fvar}{\rho [a\rightarrow b]}{\sigma,t,h,s} \\
		\semExt{H_{h'}\:\fvar}{\rho}{\sigma, t, h, s}&= \left\{ 
		\begin{array}{ll} 
			\bot &\text{ if } \sem{\fvar}{\rho}{\sigma,t,h',s} = \undet \text{ and } h \geq h'\\
			\sem{\fvar}{\rho}{\sigma, t, \min(h,h'), s} &\text{ otherwise }
		\end{array}
		\right. \\
		\semExt{\spaceHorizon{s'}\:\fvar}{\rho}{\sigma, t, h, s}&=  \semExt{\fvar}{\rho}{\sigma, t, h, \min(s,s')} \\
		\semExt{X}{\rho}{\sigma, t, h, s}&= \rho(X)(t,h,s)\\
		\semExt{\mu X.\: \fvar}{\rho}{\sigma, t, h, s}&= \left\{
		\begin{array}{ll}
			\lfp(f \mapsto \semExt{\fvar}{\rho [X\rightarrow f]}{\sigma,-,-,-})(t,h,s) &\text{ if } s<\infty\\
			\lfp(f \mapsto \sem{\fvar}{\rho [X\rightarrow f]}{\sigma,-,-})(t,h) &\text{ if } s = \infty
		\end{array}
		\right.\\
		\semExt{\fvar \leq \psi}{\rho}{\sigma, t, h, s}&=\bigland_{a\in A}(\sem{\fvar}{\rho}{\sigma, t, h, s}(a)\rightarrow \sem{\psi}{\rho}{\sigma, t, h, s}(a)) \\
		\semExt{\Reach{T}{D}{\fvar_1}{\fvar_2}}{\rho}{\sigma, t, h, s} &= \lfp(f \mapsto \semExt{\fvar_2 \lor (\fvar_1 \land \F{T}\diam{D} X)}{\rho [X\rightarrow f]}{\sigma,-,-,-})(t,h,s)\\
		\semExt{\Escape{T}{D}{\fvar_1}{\fvar_2}}{\rho}{\sigma, t, h, s} &= \lfp(f \mapsto \semExt{...}{\rho [X\rightarrow f]}{\sigma,-,-,-})(t,h,s)
	\end{align*}

	\section{Omitted Proofs}

	\begin{proof}[of Lemma~\ref{lem:monotony}]
		
		We prove the stated lemma by induction on \fvar:\\
		
		Base cases:
		\begin{itemize}
			\item if $\fvar = X$ is a fixpoint variable, 
			the result holds because the context is monotone,
			\item for other base cases, the semantics function
			is constant in the space horizon.
		\end{itemize}
		
		Induction step: Let $\fvar$ be a formula and assume the lemma is true for $\fvar$. 
		This means that for all $\rho, \sigma, t,h,s_1,s_2,a$ with $s_1 \leq s_2$, we have
			\[ \semExt{\fvar}{\sigma,\rho}{t,h,s_1,a} \preceq \semExt{\fvar}{\rho}{\sigma, t,h,s_2}
			\]
			
			\begin{itemize}
				\item $\land, \lor, \neg,\exists, H, \spaceHorizon{}$ preserves $\preceq$-monotonicity trivially\\
				\item $\F{T}\fvar$: case disjunction on the value of $\F{T}\fvar$: \vspace{0.5em}
				\begin{itemize}
					\item if $\semExt{\F{T}\fvar}{\sigma,\rho}{t,h,s_1,a} = \undet$: trivial
					
					\item $\semExt{\F{T}\fvar}{\sigma,\rho}{t,h,s_1,a} = \top \\\implies \exists t'\in T \: \semExt{\fvar}{\sigma,\rho}{t+t',h-t',s_1,a} = \top$.\\ By induction hypothesis,  $\semExt{\fvar}{\sigma,\rho}{t+t',h-t',s_2,a} = \top$. Therefore: $\semExt{\F{T}\fvar}{\sigma,\rho}{t,h,s_2,a} = \top$
					
					\item $\semExt{\F{T}\fvar}{\sigma,\rho}{t,h,s_1,a} = \bot\\ \implies \forall t'\in T \: \semExt{\fvar}{\sigma,\rho}{t+t',h-t',s_1,a} = \bot$. \\
					By induction hypothesis,  $\forall t'\in T \: \semExt{\fvar}{\sigma,\rho}{t+t',h-t',s_2,a} = \bot$. Therefore, $\semExt{\F{T}\fvar}{\rho}{\sigma,t,h,s_2}(a) = \bot$\\
				\end{itemize}
				
				\item $\diam{D}\fvar(a)$: case disjunction on the value of $\diam{D}\fvar(a)$ for any $a \in A$:\vspace{.5em}
				\begin{itemize}
					\item if $\semExt{\diam{D}\fvar}{\rho}{\sigma, t,h,s_1}(a) = \undet$, trivial.
					\item if $\semExt{\diam{D}\fvar}{\sigma,\rho}{t,h,s_1,a} = \top$, then 1) $h\geq 0$ and 
					there is $b\in\agents$ such that 
					2) $\sigma(t,a,b)\in D$, 
					3) $\sigma(t,a,b)\leq s_1$ 
					and 4) $\semExt{\fvar}{\sigma,\rho}{t,h,s_1 - \sigma(t,a,b),b} = \top$.
					Since $s_1 \leq s_2$, 3) implies 3') $\sigma(t,a,b)\leq s_2$.
					The induction hypothesis together with 4) implies 
					4') $\semExt{\fvar}{\sigma,\rho}{t,h,s_2 - \sigma(t,a,b),b} = \top$. 
					Altogether, 1+2+3'+4' implies $\semExt{\diam{D}\fvar}{\sigma,\rho}{t,h,s_2,a} = \top$
					
					\item if $\semExt{\diam{D}\fvar}{\sigma,\rho}{t,h,s_1,a} = \bot$, then $h\geq 0$ and
					for all $ b \in A $, one of the following statements holds:
					\begin{align}
						&\sigma(t,a,b) > \max(D) \land s_1 \geq \max(D)\\
						&s_1 \geq \sigma(t,a,b) \land \min(D) > \sigma(t,a,b)\\
						&s_1 \geq \sigma(t,a,b)\land \max(D) \in D
						\land \semExt{\fvar}{\sigma,\rho}{t,h,s_1 - \sigma(t,a,b),b} = \bot
					\end{align}
					(1) and (2) keep the same value by substituting $s_1$ with $s_2$. 
					By induction hypothesis, $\semExt{\fvar}{\sigma,\rho}{t,h,s_2 - \sigma(t,a,b),b} = \bot$, 
					(3) keep the same value as well.  
				\end{itemize}
				\item $\mu X. \fvar$: 
				By Cousot-Cousot theorem, we know that 
				the least fixpoint is obtained by 
				transfinite applications of the functional 
				we are taking the least fixpoint of, 
				starting with the constant function equal to 
				$\bot$. To show that the least fixpoint 
				is monotone in the space horizon, it is 
				enough to show the following properties:
				\begin{itemize}
					\item (base case) the constant function 
					is monotone
					\item (successor case) if $f$ is function 
					monotone in the space horizon, then 
					$\semextfun{\fvar}{\sigma,\rho[X\to f]}$
					is also monotone in the space horizon.
					This is just by induction hypothesis, with 
					the fact that $\rho[X\to f]$ is a monotone 
					context.
					\item (limit case) the supremum of a chain 
					of functions monotone in the space horizon 
					is monotone too. Here, we have to be careful,
					because the orders on $\threeval$ used for the supremum
					($\bot\leq\undet\leq \top$) is not the same 
					as the one used for the monotonicity (information order).
					Let $(f_i)_{i\in I}$ be a chain for the pointwise 
					extension of $\leq$ where each $f_i$ is monotone 
					in the space horizon for the information order, $f$ its supremum and 
					fix $t,h,s_1\leq s_2, a$.
					Since $(f_i)_{i\in I}$ is a chain for the pointwise order, 
					then $(f_i(t,h,s_1,a))_{i\in I}$
					and $(f_i(t,h,s_2,a))_{i\in I}$ are chains for 
					$\leq$. Since $\threeval$ is finite, 
					those chains are ultimately constant, that 
					is there are $i_1$ and $i_2$ such that 
					for all $i\geq i_j$, 
					$f_i(t,h,s_j,a)=f_{i_j}(t,h,s_j,a)$.
					Then, if $i=\max\{i_1,i_2\}$, 
					$f(t,h,s_1,a)=f_i(t,h,s_1,a)\preceq f_i(t,h,s_2,a)=f(t,h,s_2,a)$.
					\qed
				\end{itemize}
			\end{itemize}
	\end{proof}

	\begin{proof}[of Lemma~\ref{lem:vanilla-infinity}]
		We proceed by induction on the formula.
		The base cases are trivial, as the space horizon value has no effect on the value.
		
		For the inductive cases, we assume that the lemma holds for a formula $\fvar$. The non-trivial cases are the following:
		\begin{itemize}
			\item $\diam{D}$: 
			We do a case disjunction on the value of $\diam{[d_1,d_2]} \fvar $ for agent $a$:
			\begin{itemize}
				\item $\semExt{\diam{[d_1,d_2]} \fvar}{\rho}{\sigma,t,h,\infty}(a) = \top \\ \iff \exists b \in A \: d \in [d_1,d_2] \land \semExt{\fvar}{\rho}{\sigma,t,h,\infty}(b) = \top \\ \iff \exists b \in A \: d \in [d_1,d_2] \land \sem{\fvar}{\rho}{\sigma,t,h}(b) = \top \\ \iff \sem{\diam{[d_1,d_2]} \fvar}{\rho}{\sigma,t,h}(a) = \top $ 
				\item $\semExt{\diam{[d_1,d_2]} \fvar}{\rho}{\sigma,t,h,\infty} = \bot \\ \iff \forall b\in A \: d \in [d_1,d_2] \lor \semExt{\fvar}{\rho}{\sigma,t,h,\infty}(b) = \bot \\
				\iff \forall b \in A : d \in [d_1,d_2] \lor \semExt{\fvar}{\rho}{\sigma,t,h}(b) = \bot \\
				\iff \sem{\diam{[d_1,d_2]} \fvar}{\rho}{\sigma,t,h}(a) = \bot
				$
				\item if $\semExt{\diam{[d_1,d_2]} \fvar}{\rho}{\sigma,t,h,\infty}(a) = \undet \\
				\iff \forall b \in A \:  d \in [d_1,d_2] \land \semExt{\fvar}{\rho}{\sigma,t,h,\infty}(b) = \undet
				\\
				\iff \forall b \in A \:  d \in [d_1,d_2] \land \semExt{\fvar}{\rho}{\sigma,t,h}(b) = \undet \\
				\iff \sem{\diam{[d_1,d_2]} \fvar}{\rho}{\sigma,t,h}(a) = \undet
				$
			\end{itemize}
		\end{itemize}

		\begin{itemize}
			\item $\mu X .\fvar$: obvious since we use 
			the pointwise order on functions.
		\end{itemize}
		
	\end{proof}

	\begin{proof}[of Lemma~\ref{lem:restriction}]
		Once again, the proof proceeds by induction on the formula \fvar.
		The base cases are again trivial, as the space horizon parameter does not influence their value.
		For the inductive cases, we assume that the lemma is true for a subformula $\fvar'$. 
		
		For most operators, the inductive step is straightforward:the value at a given space horizon $s$ is independent of the results at larger horizons. We detail less trivial cases below:
			
		\begin{itemize}
			
			\item Space horizon modality $\spaceHorizon{s}$: noticing that $s_{\phi}$ is the maximal value that $s$ can take makes the proof trivial.
			
			\item Diamond modality is easy as well: 
			\begin{itemize}
				\item For any finite space horizon $s \in [0,s_{\phi}]$, the valuation does not depend on any $s'$ greater than $s$.
				\item For the infinite space horizon case $s=+\infty$, the result is determined by $\rsemextfunbis{\fvar}{\sigma,\overline{\rho}}(t,h,+\infty,a)$, which, by induction hypothesis, equals $\overline{\semextfun{\fvar}{\sigma,\rho}}(t,h,+\infty,a)$.
			\end{itemize}
			
			\item The fixpoint construction is preserved under restriction. 
			Indeed, it is sufficient to say that the restriction 
			function is continuous, and, by induction 
			hypothesis, it also commutes with the abstract 
			semantic. The fixpoint is therefore preserved.\qed
			
		\end{itemize}

%
%
%
	\end{proof}

	\begin{proof}[of Lemma~\ref{lem:termination}]
	The statement about the fixpoint is obvious
	because there are only finitely many $w$ as in 
	\eqref{eq:discrete-signal}.

	For the first statement, we reason by induction 
	on $\psi$, and let's do only the interesting cases.

	Assume $\psi\equiv\spaceHorizon{s'}\:\psi'$.
	By induction hypothesis, 
	$\rsemextfunbis{\psi'}{\sigma,\rho}=\overline{w'}$
	for some 
	\[w'=\{\; (t,
		\{\; (h,
			\{\;(s, 
				\text{Uset}_{t,h,s}'
			)\;\}
		)\;\}  
	)\;\}\] 
	as in \eqref{eq:discrete-signal}.
	We define $w$ as follows. 
	First observe that $s'\in\mathcal{S}$ by 2).
	For $s<s'\in\mathcal{S}$ and every $t,h$, 
	define $\text{Uset}_{t,h,s}=\text{Uset}_{t,h,s}'$.
	Now, for every $t,h$ and every $a\in\agents$, 
	since the semantics is monotone in 
	the space horizon, 
	either there is $s_{t,h,a}\in\mathcal{S}$
	such that, for $s < s_{t,h,a}$, 
	$\text{Uset}_{t,h,s}'(a)=\undet$ and 
	for $s \geq s_{t,h,a}$, 
	$\text{Uset}_{t,h,s}'(a)\neq\undet$,
	or $\text{Uset}_{t,h,s}'(a)=\undet$ for all $s$.
	With this in mind, define for $s\geq s'$
	the uncertain set $\text{Uset}_{t,h,s}$ obtained 
	from $\text{Uset}_{t,h,s}'$ by changing  
	$\text{Uset}_{t,h,s}'(a)$ to $\bot$ if $s_{t,h,a}$
	does not exist, or if it exists and $s < s_{t,h,a}$.

	The case $\phi\equiv H_{h'}\:\psi'$ is similar.

	Assume $\psi\equiv\Diamond{D}\:\psi'$.
	It is enough to construct one discrete signal $w^b$
	as in \eqref{eq:discrete-signal}
	for each agent $b$ to correspond to 
	what we denoted 
	$\semExt{\Diamond{D}\:\psi'}{\sigma,\rho}{t,h,s,a,b}$
	since taking the supremum of such discrete 
	signals is also a discrete signal of the same shape.
	Now each of these $w^b$ can be describe as a 
	supremum (depending on the value $\top$, 
	$\undet$, or $\bot$) of a Boolean combinations 
	of such discrete signals.
	For example, since $\sigma(t,a,b)\in\mathcal{S}$ by 
	1) and 4).
	for every $t\in\mathcal{T}$ such that $\sigma(t,a,b) \geq s_\phi$,
	the condition $\sigma(t,a,b)\in D \land\sigma(t,a,b)\leq s$ can be 
	seen as the discrete signal such that 
	$\text{Uset}_{t,h,s}(a)=\top$ if 
	$\sigma(t,a,b)\geq s$ and $\bot$ otherwise.
	Similarly for $\sigma(t,a,b) > s$ and 
	$\max(D) > s$ (by 3).
	Now, the condition $\rsemextfunbis{\psi'}{\sigma,\rho}(t,h,s-\sigma(t,a,b),b)=\top$
	is obtained from the discrete signal $w'$ by induction 
	hypothesis on $\psi'$, by setting 
	$\text{Uset}_{t,h,s}(a)=\top$ if 
	$\text{Uset}'_{t,h,s-\sigma(t,a,b)}(b)=\top$, and 
	$\bot$ otherwise, observing that $s-\sigma(t,a,b)\in\mathcal{S}$ 
	if $s\in\mathcal{S}$ by 4).
	Similarly for the $\undet$ case.

	The case for $\psi=\F{T}\:\psi'$ is similar.\qed

	\end{proof}

\end{document}